\documentclass[12pt]{elsarticlerev}

\usepackage{amsmath}
\usepackage{amsthm}

\DeclareMathOperator{\ds}{\mathit{dnf\_size}}
\DeclareMathOperator{\e}{\mathit{ess}}

\DeclareMathOperator{\cs}{\mathit{cnf\_size}}
\DeclareMathOperator{\mi}{\mathit{min\_imp}}

\begin{document}
\begin{frontmatter}
\newtheorem{theorem}{Theorem}
\newtheorem{lemma}{Lemma}
\newtheorem{proposition}{Proposition}
\newtheorem{corollary}{Corollary}
\newdefinition{property}{Property}
\newdefinition{definition}{Definition}

\title{On the Gap Between $\e(f)$ and $\cs(f)$}

\author[lh] {Lisa Hellerstein}
\ead{hstein@poly.edu}
\author[dk]{Devorah Kletenik}
\ead{dkletenik@cis.poly.edu}
\address{Polytechnic Institute of NYU, 6 Metrotech Center, Brooklyn, N.Y., 11201}

\begin{abstract}
Given a Boolean function $f$, the quantity $\e(f)$ denotes the largest
set of assignments that falsify $f$, no two of which falsify a common 
implicate of $f$.  
Although $\e(f)$ is clearly a lower bound on $\cs(f)$ (the minimum
number of clauses in a CNF formula for $f$),
\u{C}epek et al.\ showed
it is not, in general,  a tight lower bound~\cite{cepek}. They 
gave examples of functions $f$ for which there is a small
gap between $\e(f)$ and $\cs(f)$.
We demonstrate significantly larger gaps.  We show that
the gap can be exponential in $n$ for
arbitrary Boolean functions, and $\Theta(\sqrt{n})$ for Horn functions,
where $n$ is the number of variables of $f$.  
We also introduce a natural extension of the quantity $\e(f)$,
which we call $\e_k(f)$, which is the largest set of assignments, no $k$ of which
falsify a common implicate of $f$.
\end{abstract}


\end{frontmatter}

\section{Introduction}

Determining the smallest CNF formula for a given Boolean function $f$ 
is a difficult problem that has been studied for many years. (See~\cite{twolevel} for an overview of relevant literature.)
Recently, \u{C}epek et al. introduced a combinatorial quantity, $\e(f)$,
which lower bounds $\cs(f)$, the minimum number of clauses in a
CNF formula representing $f$~\cite{cepek}.  
The  quantity $\e(f)$ is equal to the size of the largest set of
falsepoints of $f$, no two of which falsify the same implicate of $f$.
\footnote{This definition immediately follows from Corollary 3.2 of \u{C}epek et al.~\cite{cepek}.}

For certain subclasses of Boolean functions,
such as the monotone (i.e., positive) functions, $\e(f$) is equal to $\cs(f)$.  However,
\u{C}epek et al. demonstrated that
there can be a gap between $\e(f)$ and $\cs(f)$.
They  constructed a Boolean function $f$ on $n$
variables such that there is a multiplicative gap of size $\Theta(\log n)$
between $\cs(f)$ and $\e(f)$.\footnote {Their function is actually 
defined in terms of 
two parameters $n_1$ and $n_2$. 
Setting
them to maximize the 
multiplicative gap between $\e(f)$ and $\cs(f)$, 
as a function of the number of variables $n$,
yields a gap of size 
$\Theta(\log n)$.} 
Their constructed function $f$ is a Horn function.
Their results leave open the possibility that $\e(f)$ could be
a close approximation to $\cs(f)$. 

We show that this is not the case.
We construct a Boolean function $f$ on $n$ variables such
that there is a multiplicative gap of size $2^{\Theta{(n)}}$
between $\cs(f)$ and $\e(f)$.
Note that such a gap could not be larger than $2^{n-1}$, since 
$\cs(f) \leq 2^{n-1}$ for all functions $f$.

We also construct a Horn function $f$ such that there is a multiplicative
gap of size $\Theta(\sqrt{n})$ between 
$\cs(f)$ and $\e(f)$.
We show that no gap larger than $\Theta(n)$ is possible.

If one expresses the gaps as a function of $\cs(f)$, rather than
as a function of the number of variables $n$, then the gap we obtain
with both the constructed non-Horn and Horn functions $f$ is $\cs(f)^{1/3}$.
Clearly, no gap larger than $\cs(f)$ is possible.

We briefly explore a natural generalization of the quantity $\e(f)$, which 
we call
$\e_k(f)$, which is the largest set of falsepoints, no $k$ of which
falsify a common implicate of $f$.  The quantity $\e(f)/(k-1)$ is 
a lower bound
on CNF-size, for any $k \geq 2$.

The above results concern the size of CNF formulas.
Analogous results hold for DNF formulas by duality.

\section{Preliminaries}
\subsection{Definitions}

A Boolean function $f(x_1, \dots ,x_n)$ is a mapping ${\{0,1\}}^n \rightarrow \{0,1\}$. 
(Where it does not cause confusion, we often use the word ``function'' to refer to a Boolean function.)
A variable $x_i$ and its negation $\neg{x_i}$ are \emph{literals} (positive and negative
respectively).
A \emph{clause} is a disjunction ($\vee$) of literals.
A \emph{term} is a conjunction ($\wedge$) of literals.
A \emph{CNF} (conjunctive normal form) formula is a formula of the form
$c_0 \wedge c_1 \wedge \dots c_k$, where each $c_i$ is a clause.
A \emph{DNF} (disjunctive normal form) formula is a formula of the form
$t_0 \vee t_1 \vee \dots t_k$, where each $t_i$ is a term.

A clause $c$ containing variables from $X_n = \{x_1, \ldots, x_n\}$
is an \emph{implicate} of $f$ 
if for all $x \in \{0,1\}^n$,
if $c$ is falsified by $x$ then $f(x) = 0$.
A term $t$ containing variables from 
$X_n$
is an \emph{implicant} of function $f(x_1,\ldots,x_n)$ if for all $x \in \{0,1\}^n$,
if $t$ is satisfied by $x$ then $f(x) = 1$.
 

We define the {\em size} of a CNF formula to be the number of its clauses,
and the {\em size} of a DNF formula to be the number of its terms.

Given a Boolean function $f$,
$\cs(f)$ is the size of the smallest CNF formula representing $f$.
Analogously, $\ds(f)$ is the size of the smallest DNF formula representing $f$.

An assignment $x \in \{0,1\}^n$ is a {\em falsepoint} of $f$ if
$f(x) = 0$, and is a {\em truepoint} of
$f$ if $f(x) = 1$.
We say that a clause $c$ {\em covers} a falsepoint $x$ of $f$ if $x$ falsifies $c$. A term $t$ {\em covers} a truepoint $x$ of $f$ if $x$ satisfies $t$.

A CNF formula $\phi$ representing a function $f$ forms a {\em cover} of the falsepoints of $f$,
in that each falsepoint of $f$ must be covered by at least one clause of $\phi$.
Further, if $x$ is a truepoint of $f$, then no clause of $\phi$ covers $x$.
Similarly, a DNF formula $\phi$ representing a function $f$ forms a {\em cover} of the truepoints of $f$, 
in that each truepoint of $f$ must be covered by at least one term of $\phi$.
Further, if $x$ is a falsepoint of $f$, then no term of $\phi$ covers $x$.

Given two assignments $x,y \in \{0,1\}^n$, we write
$x  \leq y$ if $\forall i, x_i \leq y_i$. An assignment $r$ \emph{separates} two assignments $p$ and $q$ if $\forall i$, $p_i = r_i$ or $q_i = r_i$.

A \emph{partial} function $f$ maps ${\{0,1\}}^n$ 
to  $\{0,1, \ast\}$, where $\ast$ indicates that the value of $f$ is not defined on the assignment. A Boolean formula $\phi$ is 
\emph{consistent} with a partial function $f$ 
if $\phi(a) = f(a)$ for all $a \in \{0,1\}^n$ where $f(a) \neq \ast$. If $f$ is a partial Boolean function, then $\cs(f)$ and $\ds(f)$ are the
size of the smallest CNF and DNF formulas consistent with the $f$, respectively.

A Boolean function $f(x_1, \ldots, x_n)$ is \emph{monotone} if
for all $x,y \in \{0,1\}^n$, if $x \leq y$ then
$f(x) \leq f(y)$. 
A Boolean function is \emph{anti-monotone} if for all $x,y \in \{0,1\}^n$, if $x \geq y$ then
$f(x) \leq f(y)$. 

A DNF or CNF formula is \emph{monotone} if it contains no negations; it is anti-monotone if
all variables in it are negated.
A CNF formula is a Horn-CNF if each clause contains at most one variable without a negation.
If each clause contains exactly one variable without a negation it is a \emph{pure} Horn-CNF.
A \emph{Horn function} is a Boolean function that can be represented by a Horn-CNF.
It is a 
\emph{pure Horn function} if it can be represented by a pure Horn-CNF. 
Horn functions are a generalization of anti-monotone
functions, and have applications in artficial intelligence~\cite{norvig}.

We say that two falsepoints, $x$ and $y$, of a function $f$ are \emph{independent} if no implicate of
$f$ covers both $x$ and $y$.
Similarly, we say that two truepoints $x$ and $y$ of a function $f$ are \emph{independent} if no implicant of
$f$ covers both $x$ and $y$.  
We say that a set $S$ of falsepoints (truepoints) of $f$ is independent if all pairs of falsepoints (truepoints) in $S$ are independent.


%
%

The {\em set covering problem} is as follows: Given a ground set $A = \{e_1, \ldots, e_m\}$ of
elements, a set ${\cal S} = \{S_1, \ldots, S_n\}$ of subsets of $A$, and a positive integer $k$,
does there exist ${\cal S}' \subseteq {\cal S}$ such that $\bigcup_{S_i \in {\cal S}'} = {\cal S}$ and $|{\cal S}'| \leq k$?
Each set $S_i \in {\cal S}$ is said to cover the elements it contains.  Thus the set covering
problem asks whether
$A$ has a ``cover'' of size at most $k$.

A set covering instance is $r$-uniform, for some $r > 0$, if
all subsets $S_i \in {\cal S}$ have size $r$.

Given an instance of the set covering problem, we say that a 
subset $A'$ of ground set $A$
is \emph{independent} if no two elements of $A'$ are contained in a common subset $S_i$ of ${\cal S}$.


\section{The quantity $\e(f)$}

We begin by restating the definition of $\e(f)$ in terms of independent falsepoints. We also introduce an analogous
quantity for truepoints. (The notation $\e^d$ refers to the fact that this is a dual definition.)

\begin{definition}
  Let $f$ be a Boolean function.  The quantity $\e(f)$ denotes the
size of the largest independent set of falsepoints of $f$.  The quantity $\e^d(f)$ denotes
the largest independent set of truepoints of $f$.
\end{definition}

As was stated above,
\u{C}epek et al. introduced the quantity $\e(f)$ as a lower bound on $\cs(f)$.
The fact that $\e(f) \leq \cs(f)$ follows easily from the above definitions, and
from the following facts: (1) if $\phi$ is a CNF formula representing $f$, then every falsepoint of $f$ must
be covered by some clause of $\phi$, and (2) each clause of $\phi$ must be an 
implicate of $f$.

Let $f'$ denote the function that is the complement of $f$, i.e. $f'(a) = \neg f(a)$ for all assignments
$a$. Since, by duality, $\e(f') = \e^d(f)$ and $\cs(f') = \ds(f)$, it
follows that $\e(f') \leq \ds(f)$.

\begin{property}
\cite{cepek} Two falsepoints of $f$,  $x$ and $y$, are independent iff there exists a truepoint $a$ of $f$
that separates $x$ and $y$.
\end{property}





Consider the following decision problem, which we will call \emph{ESS}: 
``Given a CNF formula representing a Boolean function $f$,
and a number $k$, is $\e(f) \leq k$?''
Using Property 1,
this problem is easily shown to be in co-NP~\cite{cepek}.


We can combine 
the fact that \emph{ESS} is in co-NP with results on the hardness
of approximating CNF-minimization, 
to get the following preliminary result, based on a complexity-theoretic
assumption.

\begin{proposition}
If co-NP $\neq \Sigma_2^P$, then for some $\gamma > 0$,
there exists an infinite set of Boolean functions $f$ such that
$\e(f) n^{\gamma} < \cs(f)$, where $n$ is the number of variables of $f$.
\end{proposition}

\begin{proof}

Consider the \emph{Min-CNF} problem (decision version):
Given a CNF formula representing a Boolean function $f$,
and a number $k$, is $\cs(f) \leq k$?
Umans proved that it is $\Sigma_2^P$-complete to approximate this problem to
within a factor of $n^{\gamma}$, for some $\gamma > 0$, where $n$ is the 
number of variables of $f$~\cite{Umans99}.  
(Approximating this problem to within
some factor $q$ means answering ``yes'' whenever $\cs(f) \leq k$, and answering
``no'' whenever $\cs(f) > kq$.   If $k < \cs(f) \leq kq$, either
answer is acceptable.)

Suppose $\e(f) n^{\gamma} \geq \cs(f)$ for all Boolean functions  $f$.
Then one can approximate Min-CNF to within a factor of $n^{\gamma}$
in co-NP by simply using the co-NP algorithm for \emph{ESS} to determine whether
$\e(f) \leq k$.  
Even if $\e(f) n^{\gamma}\geq \cs(f)$ for a finite set $S$ of functions,
one can still approximate Min-CNF to within a factor of $n^{\gamma}$
in co-NP, by simply handling the finite
number of functions in $S$ explicitly as special cases.
Since approximating \emph{Min-CNF} to within this factor is $\Sigma_2^P$-complete,
$\Sigma_2^P \subseteq$ co-NP.  
By definition, co-NP$ \subseteq \Sigma_2^P$, so $\Sigma_2^P =$ co-NP.  
\end{proof}

The non-approximability result of Umans for \emph{Min-CNF}, used in the above
proof, is expressed in terms of the number of variables $n$ of the function.
Umans also showed~\cite{Umans98} that it is $\Sigma_2^P$ complete to approximate \emph{Min-CNF} to within
a factor of $m^{\gamma}$, for some $\gamma \geq 0$, where
$m = \cs(f)$.  Thus we can also prove that,
if NP $\neq$ $\Sigma_2^P$, then for some $\gamma > 0$,
there is an infinite set 
of functions $f$ such that $\e(f) < \cs(f)^{1-\gamma}$.

The assumption that $\Sigma_2^P \neq$ co-NP is not unreasonable, so 
we have grounds to believe that there is an infinite set of functions
for which the gap between $\e(f)$ and $\cs(f)$ is greater than $n^\gamma$
(or $\cs(f)^{\gamma}$) for some $\gamma$.
Below, we will explicitly construct such sets
with larger gaps than that of Proposition 1, and with no complexity theoretic
assumptions.

We can also prove a proposition similar to Proposition 1
for Horn functions, using a different complexity theoretic assumption.
(Since the statement of the proposition includes
a complexity class parameterized by the standard input-size parameter $n$,
we use $N$ instead of $n$ to denote the number of inputs to a Boolean
function.)

\begin{proposition}
If NP $\not\subseteq$ co-NTIME($n^{polylog(n)})$, then for some $\epsilon$
such that $0 < \epsilon < 1$,  there exists 
an infinite set of Horn functions $f$ such that
$\frac{\cs(f)}{\e(f)} \geq 2^{\log^{ 1-\epsilon} N}$, where $N$ is the number
of input variables of $f$.
\end{proposition}

\begin{proof}
Consider the following \emph{Min-Horn-CNF} problem (decision version):
Given a Horn-CNF $\phi$ representing a Horn function $f$, and an integer $k \geq 0$, 
is $\cs(f) \leq k$?
Bhattacharya et al. \cite{Bhattacharya_onapproximate} showed that there exists
a deterministic,
many-one reduction (i.e. a Karp reduction), running in time $O(n^{polylog(n)})$ (where $n$ is
the size of the input), from an NP-complete problem to the problem of approximating \emph{Min-Horn-CNF}
to within a factor of $2^{\log^{ 1-\epsilon} N}$, where $N$ is the number of input variables of $f$ .

Suppose that $\frac{\cs(f)}{\e(f)}$ is at most $2^{\log^{ 1-\epsilon} N}$ for all 
Boolean functions $f$. 
It is well known that given a Horn-CNF $f$, the size of the smallest
(functionally) equivalent Horn-CNF is precisely $\cs(f)$.
Thus given a Horn-CNF $\phi$ on $N$
variables, and a number $k$, if there does not exist a 
Horn-CNF equivalent to $\phi$ of size
less than $2^{\log^{ 1-\epsilon} N} \times k$, this can be verified 
non-deterministically in polynomial time (by verifying that $\e(f) \geq k$).
Thus the complement of \emph{Min-Horn-CNF} is approximable to within a 
factor of $2^{\log^{ 1-\epsilon} N}$, in deterministic time $n^{polylog(n)}$
(where $n$ is the size in bits of the input Horn-CNF, and
$N$ is the number of variables in the input Horn-CNF).
Combining this fact with the reduction of Bhattacharya et al.
implies that the complement of an NP-complete problem
can be solved in non-deterministic time $n^{polylog(n)}$.
Thus NP is contained in co-NTIME($n^{polylog(n)}$).
The same holds if 
$\frac{\cs(f)}{\e(f)}$ is at most $2^{\log^{ 1-\epsilon} n}$ for all 
but a finite set of Boolean functions $f$.
\end{proof}

\section{Constructions of functions with large gaps between $\e(f)$ and $ \cs(f)$}

We will begin by constructing a function $f$, 
such that $\frac{\cs(f)}{\e(f)} = \Theta(n)$.  This 
is already a larger gap than the
multiplicative gap of $\log(n)$ achieved by the construction of \u{C}epek et al.~\cite{cepek}, and
the gap of $n^{\gamma}$ in Proposition 1.  We describe the
construction of $f$, prove bounds on $\cs(f)$ and $\e(f)$, and
then prove that the ratio $\frac{\cs(f)}{\e(f)}$ = $\Theta(n)$. 

We will then show how to modify this construction 
to give a function $f$ 
such that $\frac{\cs(f)}{\e(f)} = 2^{\Theta(n)}$, thus increasing the gap to be exponential in $n$.

At the end of this section, we will explore
$\e_k(f)$, our generalization of $\e(f)$.


\subsection{Constructing a function with a linear gap}
\label{linear}
\begin{theorem} 
There exists a function $f(x_1, \ldots, x_n)$
such that $\frac{\cs(f)}{\e(f)} = \Theta(n)$.
\end{theorem}

\begin{proof} 
We construct a function $f$ such that $\frac{\ds(f)}{\e^d(f)} = \Theta(n)$.
Theorem \ref{linear} then follows immediately by duality.

Our construction relies heavily on a reduction of Gimpel from
the 1960's~\cite{gimpel},
which reduces a generic instance of the set covering problem to a DNF-minimization problem.
(See Czort~\cite{czort} or Allender et al.~\cite{circuits} for more recent discussions of this reduction.) 

Gimpel's reduction is as follows.
Let $A = \{e_1, \ldots, e_m\}$ be the ground set of the set covering instance, and let ${\cal S}$ be the
set of subsets $A$ from which the cover must be formed.
With each element $e_i$ in $A$, associate a Boolean input variable $x_i$.
For each $S \in {\cal S}$, let $x_S$ denote the assignment in $\{0,1\}^m$
where $x_i = 0$ iff $e_i \in S$.
Define the partial function $f(x_1,\ldots, x_m)$ as follows:

\[
f( {x}) = \left\{ 
\begin{array}{l l}
1& \quad \mbox {if ${x}$ contains exactly $m-1$ ones} \\
\ast & \quad \mbox {if ${x} \geq x_S$ for some $S \in {\cal S}$} \\
0 & \quad \mbox {otherwise} \\
\end{array} \right.
\]

There is a DNF formula of size at most $k$ that is consistent with this
partial function if and only if the elements $e_i$ of the set covering instance $A$ can be covered
using at most $k$ subsets in ${\cal S}$ (cf. ~\cite{czort}).

We apply this reduction to the simple, 2-uniform, set covering instance over
$m$ elements where
${\cal S}$ consists of all subsets containing exactly two of those $m$ elements.
The smallest set cover for this 
instance is clearly $\lceil m/2 \rceil$.
The largest independent set of 
elements is only of size 1, since every pair of elements is contained in a common subset of ${\cal S}$.
Note that this gives a ratio of minimal set 
cover to largest independent set of $\Theta(m).$

Applying Gimpel's reduction to
this simple set covering instance, we get the following partial function $\hat{f}$:

\[
\hat{f}( x) = \left\{ 
\begin{array}{l l}
1& \quad \mbox {if $x$ contains exactly $m-1$ ones} \\
\ast & \quad \mbox {if $x$ contains exactly $m-2$ ones} \\
\ast & \quad \mbox {if $x$ contains exactly $m$ ones}  \\
0 & \quad \mbox {otherwise} \\
\end{array} \right.
\]

Since the smallest set cover for the instance has size
$\lceil m/2 \rceil$,
$$\ds(\hat{f}) = \lceil m/2 \rceil.$$

Allender et al.\ extended the reduction of Gimpel by converting
the partial function $f$ to a total function $g$.
The conversion is as follows:

Let $t = m+1$ and let $s$ be the number of $\ast$'s in $f(x)$.
Let $y_1$ and $y_2$ be two additional Boolean variables, and let 
$z$ = $z_1 \ldots z_t$ be a vector of $t$ more Boolean variables. Let 
$S \subseteq \{0,1\}^t$ be a collection of $s$ vectors, each containing an 
odd number of 1's (since $s \leq 2^m$, such a collection exists).
Let $\chi$ be the function such that
$\chi(x) = 0$ if the parity 
of $x$ is even and $\chi(x) = 1$ otherwise.

The total function $g$ is defined as follows:

\[
g( x, y_1, y_2, z) = \left\{ 
\begin{array}{l l}
1& \quad \mbox {if $f(x) = 1$ and $y_1 = y_2 = 1$ and $z \in S$} \\
1 & \quad \mbox {if $f(x) = \ast$ and $y_1 = y_2 = 1$} \\
1 & \quad \mbox {if $f(x) = \ast, y_1 = \chi(x),$ and $y_2 = \neg \chi(x)$ }\\
0 & \quad \mbox {otherwise} \\
\end{array} \right.
\]

Allender et al. proved that this total function $g$ obeys the following property:

$$\ds(g) = s(\ds(f) +1).$$

Let $\hat{g}$ be the total function obtained by setting $f = \hat{f}$ in the above definition of $g$.

We can now compute $\ds(\hat{g})$.
Let $n$ be the number of input variables of $\hat{f}$.
The total function $\hat{g}$ is defined on $n = 2m + 3$ variables. 
Since $\ds(\hat{f}) = \lceil m/2 \rceil$, we have 

\[\ds(\hat{g}) = s\left(\lceil \frac{m}{2}\rceil +1\right) \geq s\left(\frac{n-3}{4} + 1\right)\]

\noindent where $s$ is the number of assignments $x$ for which $\hat{f}(x) = *$.

We will upper bound $\e^d(\hat{g})$ by dividing the truepoints 
of $\hat{g}$ into two disjoint sets and upper-bounding the size of a maximum independent set of truepoints in each. 
(Recall that two truepoints of $\hat{g}$ are independent if they do not satisfy a common implicant of $\hat{g}$.)

\begin{itemize}
\item[Set 1:] The set of all truepoints of $\hat{g}$ whose $x$ component has the property $f(x) = \ast$. 

Let $a_1$ be a maximum independent set of
truepoints of $\hat{g}$ consisting only of points in this set.
Consider two truepoints $p$ and $q$ in this set that have the same $x$ value. 
It follows that they share the same values for $y_1$ and $y_2$. 
Let $t$ be the term containing all variables $x_i$, and exactly
one of the two $y_j$ variables,
such that each $x_i$ appears without negation if it set to 1 by $p$ and $q$, 
and with negation otherwise,  and $y_j$ is 
set to 1 by both $p$ and $q$. 
Clearly, $t$ is an implicant of $\hat{g}$ by definiton of $\hat{g}$, and clearly $t$ covers both $p$ and $q$. It follows that $p$ and $q$ are not independent.

Because any two truepoints in this set with the same $x$ value 
are not independent, $|a_1|$
cannot exceed the number of different $x$ assignments. There are $s$ assignments such that $\hat{f}(x) = \ast$, so  $|a_1| \leq s$.

\item[Set 2:] The set of all truepoints 
of $\hat{g}$ whose $x$ component has the property $\hat{f}(x) = 1$. 

Let $a_2$ be a maximum independent set consisting only of points in this set.
Consider any two truepoints $p$ and $q$ in this set that contain the same assignment for $z$. We can construct a term $t$ of the form $wy_1y_2\widetilde{z}$ such that $w$ contains exactly $m-2$ $x_i$'s that are set to 1 by both $p$ and $q$, and all $z_i$s that are set to 1 by $p$ and $q$ appear in $\widetilde{z}$ without negation, and all other $z_i$s appear with negation. It is clear that $t$ is an implicant of $\hat{g}$ and that $t$ covers both $p$ and $q$. Once again, it follows that $p$ and $q$ are not independent truepoints of $g$.

Because any two truepoints in this set with the same $z$ value are not independent, $|a_2|$ cannot exceed the number of different $z$ assignments. There are $s$ assignments to $z$ such that $z \in S$, so $|a_2| \leq s$.
\end{itemize}

Since a maximum independent set of truepoints of $\hat{g}$
can be partitioned into an
independent set of points from the first set, and an independent set of points from the second set,
it immediately follows that
\footnote{It can actually be proved that in fact, $\e^d(\hat{g}) = 2s$, but details of this proof are omitted.}

\[\e^d(\hat{g}) \leq |a_1| + |a_2 | \leq s + s = 2s.\]

Hence, the ratio between the DNF size and $\e(g)$ size is:
\[\frac{s(\frac{n-3}{4} + 1)}{2s} \geq \frac{n+1}{8} = \Theta(n)\]
\end{proof}

Note that the above function gives a class of functions satisfying
the conditions of Proposition 1, for $\gamma = 1$.

\begin{corollary}
There exists a function $f$ such that $\frac{\cs(f)}{\e(f)} \geq \cs(f)^\epsilon$ for an $\epsilon \geq$ 0.
\end{corollary}
\begin{proof} In the previous construction,
$\hat{f}(x)=\ast$ for exactly $\binom{m}{2}$ + 1 points, yielding $s$ = $\Theta(n^2)$. Hence, the DNF size is $\Theta(m^3)$, making the ratio between $\ds(\hat{g})$ and $\e^d(\hat{g})$ at least $\Theta(\ds(\hat{g})^\frac{1}{3})$. The CNF result follows by duality.
\end{proof}
\subsection{Constructing a function with an exponential gap}
\label{sec:another}

\begin{theorem}
There exists a function $f$ on $n$ variables such that $\frac{\cs(f)}{\e(f)} \geq 2^{\Theta(n)}$.
\end{theorem}
\begin{proof}

As before, we will reduce a set covering instance to a DNF-minimization
problem involving a partial Boolean function $f$.
However, here we will rely on a more general version of Gimpel's reduction,
due to Allender et al., described in the following lemma.

\begin{lemma}
\cite{circuits}
\label{lemma:generalizeG}
Let ${\cal S}=\{S_1, \ldots, S_p\}$ be a set of subsets of ground
set $A = \{e_1, \dots ,e_m\}$. Let $t > 0$ and let $V = \{v^i: i \in \{1, \dots ,m\}\}$ and $W = \{w^j: j \in \{1, \dots ,p\}\}$ be sets of vectors from $\{0,1\}^t$ such that for
all $j \in \{1, \dots ,p\}$ and $i \in \{1, \dots ,m\}$, 
 $$e_i \in S_j \mbox{ iff } v^i \geq w^j$$

Let $f:\{0,1\}^t \rightarrow \{0,1,\ast\}$ be the partial function 
such that
\[
\hspace{5 mm}f( {x}) = \left\{ 
\begin{array}{l l}
1& \quad \mbox {if ${x} \in V$} \\
\ast & \quad \mbox {if ${x} \geq w$ for some $w \in W$ and $x \notin V$} \\
0 & \quad \mbox {otherwise} \\
\end{array} \right.
\]
  
Then ${\cal S}$ has a minimum cover of size $k$ iff $\ds(f) = k$.

\end{lemma}

(Note that the construction in the above lemma is equivalent to Gimpel's if
we take $t=m$, 
$V = \{v \in \{0,1\}^m | v$ contains exactly $m-1$ 1's $\}$,
and $W  = \{x_S | S \in {\cal S}\}$, where $x_S$ denotes
the assignment in $\{0,1\}^m$ where $x_i = 0$ iff $e_i \in S$.)

As before, we use the simple 2-uniform set covering instance
over $m$ elements where ${\cal S}$
consists of all subsets of two of those elements. 
The next step is to construct sets $V$ and $W$
satisfying the properties in the above lemma for this set
covering instance.  To do this, we use a randomized construction
of Allender et al.\ that generates sets $V$ and $W$
from an $r$-uniform set-covering instance, for any $r > 0$.
This randomized construction appears in the appendix of
~\cite{circuits}, and is described in the following lemma.

\begin{lemma}
\label{random}
Let $r > 0$ and let
${\cal S}=\{S_1, \ldots, S_p\}$ be a set of subsets of $\{e_1, \dots ,e_m\}$,
where each $S_i$ contains exactly $r$ elements.
Let $t \geq 3r(1+\ln(pm))$. 
Let $V = \{v^1, \ldots, v^m\}$ be a set of $m$ vectors of length $t$, where
each $v^i \in V$
is produced by randomly and independently setting each bit
of $v^i$ to 0 with probability $1/r$.
Let $W = \{w^1, \ldots, w^p\}$, where each
$w^j$ = the bitwise AND of all $v^i$ such that $e_i \in S_j$.
Then, the following holds with probability greater than 1/2:
For all $j \in \{1, \dots ,p\}$ and $i \in \{1, \dots ,m\}$, $e_i \in S_j$ iff $v^i \geq w^j$.
\end{lemma} 

By Lemma~\ref{random}, there exist
sets $V$ and $W$, each consisting of vectors of length $6(1+\ln(m^2(m-2)/2)) = O(\log m)$,
satisfying the conditions of Lemma~\ref{lemma:generalizeG} for our simple 2-uniform set covering instance.
Let $\tilde{f}$ be the partial function on $O(\log m)$ variables obtained
by using these $V$ and $W$ in the definition of $f$ in Lemma~\ref{lemma:generalizeG},

The DNF-size of $\tilde{f}$ is the size of the smallest set cover, 
which is $\lceil m/2\rceil$, and the number of variables $n$ = $\Theta$(log $m$); 
hence the DNF size is $2^{\Theta (n)}$.

We can convert the partial function $\tilde{f}(x)$ to a total function $\tilde{g}(x)$ just as done in the previous section.
The arguments regarding DNF-size and 
$\e^d(\tilde{g})$ remain the same. Hence, the DNF-size is now 
$s\left(2^{\Theta(n)} + 1\right),$ and
$\e^d(\tilde{g})$ is again at most $2s$.

The ratio between the DNF-size and $\e^d(\tilde{g})$ is therefore at least $ 2^{\Theta(n)}$. Once again, the CNF result follows.
\end{proof}

\subsection{The quantity $\e_k(f)$}
\label{sec:ess_k(f)}

We say that a set $S$ of falsepoints (truepoints) of $f$ is a 
``$k$-independent set'' if no $k$ of the falsepoints (truepoints) of $f$
can be covered by the same implicate (implicate) of $f$.

We define $\e_k(f)$ to be the size of the largest $k$-independent
set of falsepoints of $f$, and $\e_k^d(f)$ to be the size
of the largest $k$-independent set of truepoints of $f$.

If $S$ is a $k$-independent set of falsepoints of $f$, 
then each implicate of $f$ 
can cover at most $k-1$ falsepoints in $S$.
We thus have the following lower-bound on CNF-size:
$\cs(f) \geq \frac{\e_k(f)}{k-1}$.   

Like $ess(f)$, this lower bound is not tight.

\begin{theorem}
For any arbitrary $2 \leq k \leq h(n)$, where $h(n) = \Theta(n)$, there exists a function $f$ on $n$ variables, such that the gap between $\cs(f)$ and $\frac{\e_k(f)}{k-1} $ is at least $2^{\Theta(\frac{n}{k})}$.
\end{theorem}
\begin{proof}

Consider the $k$-uniform set cover instance consisting of all subsets of $\{e_1, \dots, e_m\}$ of size $k$. Construct $V$ and $W$ randomly using the construction from the appendix of \cite{circuits} described in Lemma~\ref{random}, and define a corresponding partial function $\tilde{f}$, as in Lemma 1. Note that according to the definition of $\tilde{f}$, there can be no $k$ $v^i$ for any $k$ values of $i \in \{1, \dots ,m\}$, such that all $v^i$ $\geq w^j$ for some $j \in \{1, \dots ,p\}.$ The maximum size $k$-independent set of truepoints
of $\tilde{f}$ consists of $k-1$ truepoints.

We can convert the partial function $\tilde{f}$ to a total function $\tilde{g}$ according to the construction detailed in Section \ref{linear}. Once again, we introduce $s$ new truepoints such that $\tilde{f}(x) = \ast$, yielding a maximum of $s$ pairwise independent truepoints. The definiton of $k$-independence, however, allows $k-1$ ``copies'' of these truepoints that differ in the assignments  to $z$ for each of the $s$ points. Hence, the largest $k$-independent set 
of these points can contain a maximum of $s(k-1)$ points. 

We have previously mentioned that there exist $k-1$ $k$-independent ground elements (i.e., $\tilde{f}(x) = 1$ truepoints). Once again, when we consider the $s$ $ \tilde{z}$ portion of the term, where no two $\tilde{z}$ portions can be covered by the prime implicate, we can include a total of $s(k-1)$ of these truepoints.  Hence, the largest independent set for points of this type is of size is of size no greater than $s(k-1)$. Since these two types of truepoints are independent, $\e_k^d(\tilde{g}) \leq 2s(k-1)$.

The lower bound on DNF size, $\frac{\e_k^d(f)}{k-1}$, is, for this $\tilde{g}$, $ \leq \frac{2s(k-1)}{k-1} \leq 2s$. The ratio between that and the actual DNF size is 
\[\frac{s(2^{\Theta(\frac{n}{k})}+1)}{2s}  \geq 2^{\Theta(\frac{n}{k})}.\]

The CNF result clearly follows. \end{proof}

\section{Size of the gap for Horn Functions}
\label{sec:Horn}

Because Horn-CNFs contain at most one unnegated variable per clause, they can be expressed as implications; eg. $\bar{a} \vee b$ is equivalent to $a \rightarrow b$. Moreover, a conjunction of several clauses that have the same antecedent can be represented as a single \emph{meta-clause}, where the antecedent is the antecedent common to all the clauses and the consequent is comprised of a conjunction of all the consequents, eg. ($a \rightarrow b) \wedge (a \rightarrow c)$ can be represented as $a \rightarrow (b \wedge c)$.

\subsection{Bounds on the ratio between $\cs(f)$ and $\e(f)$}

Angluin, Frazier and Pitt \cite{angluin} presented an algorithm (henceforth: the AFP algorithm) to learn Horn-CNFs, where the output is a series of meta-clauses. It can be proven \cite{ariaslattice, ariasgdNew} that the output of the algorithm is of minimum implication size (henceforth: $\mi(f)$) -- that is, it contains the fewest number of meta-clauses needed to represent function $f$. Each meta-clause can be a conjunction of at most $n$ clauses; hence, each implication is equivalent to the conjunction of at most $n$ clauses. Therefore,
\[\cs(f) \leq n \times \mi(f).\]
The learning algorithm maintains a list of negative and positive examples (falsepoints and truepoints of the Horn function, respectively), containing at most $\mi(f)$ examples of each.

\begin{lemma}
\label{lemmaHorn}
The set of negative examples maintained by the AFP algorithm is an independent set.
\end{lemma}

\begin{proof}
The proof for this lemma relies heavily on \cite{ariaslattice}; see there for further details.

Let us consider any two negative examples, $n_i$ and $n_j$, maintained by the algorithm. There are two possibilities:
\begin{enumerate}
	\item $n_i \leq n_j$ or $n_j \leq n_i.$ (These two examples are comparable points; one is below the other on the Boolean lattice.)
	\item $n_i$ and $n_j$ are incomparable points (Neither is below the other on the lattice).
\end{enumerate}
Let us consider the first type of points: Without loss of generality, assume that $n_i \leq n_j$. Arias et al. define a positive example $n_i^\ast$ for each negative example $n_i$. This example $n_i^\ast$ has several unique properties; amongst them, that $n_i < n_i^\ast$ for all negative examples $n_i$ (Section 3 in \cite{ariaslattice}). They further prove (Lemma 6 in  \cite{ariaslattice}) that if $n_i \leq n_j$, then $n_i^{\ast} \leq n_j$ as well. Hence, any attempt to falsify both falsepoints, $n_i$ and $n_j$, with a common implicate of the Horn function would falsify the positive example ($n_i^{\ast}$) that lies between them as well. Therefore, these two points are independent.

Now let us assume that $n_i$ and $n_j$ are incomparable. Any implicate that falsifies both points is composed of variables on which the two points agree. Clearly, this implicate would likewise cover a point that is the componentwise intersection  of $n_i$ and $n_j$. However, Arias et al. prove (Lemma 7 in \cite{ariaslattice}) that $n_i \wedge n_j$ is a positive point if $n_i$ and $n_j$ are incomparable. Hence, any implicate that falsifies both $n_i$ and $n_j$ would likewise falsify the truepoint $n_i \wedge n_j$ that lies between them. Therefore, these two points cannot be falsified by the same implicate and they are independent.
\end{proof}
\begin{theorem} 
For any Horn function $f$, $\frac{\cs(f)}{\e(f)} \leq n$
\end{theorem}

\begin{proof}

For any Horn function $f$, there exists a set of negative examples of size at most $min\_imp(f)$, and these examples are all independent. Hence, $\e(f) \geq \mi(f)$. We have already stated that $\mi(f)$ is at most a factor of $n$ times larger than the minimum CNF size for this function.

Hence, $\cs(f) \leq n \times \e(f)$.

Moreover, since Lemma \ref{lemmaHorn} holds for general Horn functions in addition to pure Horn \cite{ariasgdNew}, this bound  holds for all Horn functions. 
\end{proof}

\subsection{Constructing a Horn function with a large gap between $\e(f)$ and $\cs(f)$ }
\begin{theorem}
There exists a definite Horn function $f$ on $n$ variables such that $\frac{\cs(f)}{\e(f)} \geq \Theta(\sqrt{n})$.
\end{theorem}
\begin{proof}
Consider the 2-uniform set covering instance over $k$ elements consisting of all subsets of two elements. We can construct a definite Horn formula $\varphi$ corresponding to this set covering according to the construction in \cite{bookDraft}, with modifications based on \cite{Bhattacharya_onapproximate}.

The formula $\varphi$ will contain 3 types of variables: 
\begin{itemize}
\item Element variables: There is a variable $x$ for each of the $k$ elements.
\item Set variables: There is a variable $s$ for each of the $\binom{k}{2}$ subsets.
\item Amplification variables: There are $t$ variables $z_1 \dots z_t$. 
\end{itemize}
The clauses in $\varphi$ fall into the following 3 groups:
\begin{itemize}
\item Witness clauses: There is a clause $s_j \rightarrow x_i$ for each subset and for each element that the subset covers. There are 2$\binom{k}{2}$ such clauses. 
\item Feedback clauses : There is a clause $x_1 \dots x_k \rightarrow s_j$ for each subset. There are $\binom{k}{2}$ such clauses.
\item Amplification clauses: There is a clause $z_h \rightarrow s_j$ for every $h \in \{1 \dots t\}$ and for every subset. There are $t{\binom{k}{2}}$ such clauses. 
\end{itemize}

It follows from \cite{bookDraft} that any minimum CNF for this function must contain all witness and feedback clauses, along with $tc$ amplification clauses, where $c$ is the size of the smallest set cover.

This particular function $f$ has a minimum set cover of size $k/2$; hence, $\cs(f) = 2 {\binom{k}{2}} + {\binom{k}{2}} + t(k/2).$

\vspace{5 mm}

We will upper bound $\e(f)$ by dividing the falsepoints of $f$ into three disjoint sets and finding the maximum independent set for each. 
\begin{itemize}
	\item[Set 1:] The set of all falsepoints of $f$ that contain at least one $x_i$ = 0 for $i \in \{1, \dots ,k\}$ and some $s_j$ = 1 for a subset $s_j$ that covers $x_i$. 

Let $a_1$ be the largest independent set of $f$ consisting of points in this set.
These points can be covered by an implicates of the form $s_j \rightarrow x_i$, of which there are $2{\binom{k}{2}}$. We will define the function $f'$ whose falsepoints are just the Type 1 points. Since these points are covered by the $s_j \rightarrow x_i$ implicates, $\cs(f')$ is no more than the number of $s_j \rightarrow x_i$ implicates. We have earlier said that $\e(f') \leq \cs(f')$, hence it follows that $\e(f') \leq 2{\binom{k}{2}}$. $\e(f')$ is precisely the size of $a_1$; hence, $a_1$ can contain no more than $2{\binom{n}{2}}$ points. 
		
	\item[Set 2:] The set of all falsepoints that are not in the first set, have $x_i = 1$ for all $i \in \{1, \dots ,k\}$, and at least one $s_j$ = 0 for some $j \in \{1, \dots ,{\binom{k}{2}}\}$. 

Let $a_2$ of $f$ be the largest independent set consisting of points in this set.
These points can be covered by implicates of the form $x_1 \dots x_k \rightarrow s_j$. There are $\binom{k}{2}$ such implicates. Hence, by the same argument as above, $a_2$ can contain no more than $\binom{k}{2}$ points.
	
	\item[Set 3:] The set of all falsepoints that are not in the first two sets, and therefore have $z_h=1$ for some $h \in \{1, \dots ,t\}$, $x_i$ = 0 for some $i \in \{1, \dots ,k\}$, and $y_j = 0$ for all subsets $y_j$ covering $x_i$. 

Let $a_3$ be the largest independent set of $f$ consisting of points in this set. Let us fix $h = 1$.
Consider a falsepoint $p$ in this set where $x_i = 0$ for at least one $i \in \{1, \dots ,k\}.$ If $p$ contained a $y_j = 1$ such that the subset $y_j$ covers $x_i$, that point would be a point in the first set. Hence, the only points of this form in this set have $y_j = 0$ for \emph{all} $k-1$ subsets $y_j$ that cover $x_i$. 

Now consider another falsepoint $q$ in this set, where $x_a = 0$ for at least one $a \in \{1, \dots ,k\}$. Once again, the only points in this set must set $y_b = 0$ for all $k - 1$ subsets $y_b$ that cover $x_a$.

Because the set covering problem included a set for each pair of $x_i$ points, there exists some $y_j$ that covers both $x_i$ and $x_a$. By the previous argument, that $y_j$ is set to 0 in all assignments that set $x_i$ or $x_a$ = 0. For a fixed $k$, all of these points can be covered by the implicate $z_k \rightarrow y_j$. Hence, points $p$ and $q$ are not independent.

In fact, any two falsepoints chosen that are not in the first set and contain $z_h = 1$ for the same $h$ and at least one $x_i = 0$ are not independent. Because there are $t$ values of $h$, $a_3$ therefore has size $t$. 
\end {itemize}

The largest independent set for all falsepoints cannot exceed the sum of the independent sets for these three disjoint sets, hence 
\[\e(f) \leq |a_1| + |a_2| + |a_3 | \leq 2{\binom{k}{2}} + {\binom{k}{2}} + t.\]

The gap  between $\cs(f)$ and $\e(f)$ = 
\[\frac{\cs(f)}{\e(f)} \geq \frac{3{\binom{k}{2}} + t(k/2)}{3{\binom{k}{2} + t}}.\]
Let us set $t = 3 {\binom{k}{2}}$. The difference is now:
\[\frac{\cs(f)}{\e(f)} \geq \frac{t(1 + k/2)}{2t }\geq \Theta(k).\]
We have $k$ element variables, $\binom{k}{2}$ set variables, and $3 {\binom{k}{2}}$ amplification variables, yielding $n = \Theta(k^2)$ variables in total. The difference between $\cs(f)$ and $\e(f)$ is therefore $\geq \Theta(\sqrt {n})$.
\end{proof}

We earlier posited that if $\Sigma^2_p \neq$ \emph{co-NP}, there exists an infinite set of functions for which $\frac{\cs(f)}{\e(f)} \geq \cs(f)^\gamma$ for some $\gamma > 0.$ We can now prove a stronger theorem:
\begin{theorem}
There exists an infinite set of Horn functions $f$ for which $\frac{\cs(f)}{\e(f)} \geq \cs(f)^\gamma$.
\end{theorem}
\begin{proof}
See construction above. Because $\cs(f) = \Theta(k^3)$, $\frac{\cs(f)}{\e(f)} = \Theta(\cs(f)^{1/3}).$
\end{proof}

\section{Acknowledgements}
This work was partially supported by the US Department of Education
GAANN grant P200A090157, and by NSF Grant CCF-0917153.
\bibliographystyle{model1-num-names}\
\bibliography{hellersteinKletenik2}

\end{document}